\documentclass[a4paper]{amsart}
\usepackage{amsmath,amsfonts,euscript,amssymb,amsthm,graphicx,subfigure,verbatim}
\usepackage{color}
\usepackage{hyperref}
\date{\today}

\newtheorem*{theorem}{Theorem}
\newtheorem{corollary}{Corollary}
\newtheorem{lemma}{Lemma}

\newcommand{\R}{\mathbb{R}}
\newcommand{\N}{\mathbb{N}}
\newcommand{\Z}{\mathbb{Z}}

\newcommand{\St}{\mathbb{S}^2}

\newcommand{\dd}[2]{\frac{\partial #1}{\partial #2}}

\newcommand{\m}{\boldsymbol{m}}

\newcommand{\y}{\boldsymbol{y}}
\newcommand{\eps}{\varepsilon}
\newcommand{\del}{\partial}

\newcommand{\bs}[1]{ \boldsymbol{#1}}
\newcommand{\taf}{ \boldsymbol{\tau} }
\newcommand{\e}{\boldsymbol{\hat{e}}}
\newcommand{\n}{\boldsymbol{\nu}}
\newcommand{\intt}{\int_{\R^2} }
\newcommand{\weak}{\rightharpoonup}

\renewcommand{\div}{\mathop{\mathrm{div}}\nolimits}
\newcommand{\curl}{\mathop{\mathrm{curl}}\nolimits}

\def\XXint#1#2#3{{\setbox0=\hbox{$#1{#2#3}{\int}$}
\vcenter{\hbox{$#2#3$}}\kern-.5\wd0}}

\title{Curvature stabilized skyrmions with angular momentum}

\thanks{This work is supported by Deutsche Forschungsgemeinschaft (DFG grant no. ME 2273/3-1). }

\author{Christof Melcher}
\address{RWTH Aachen University \\ Lehrstuhl I f\"ur Mathematik  \\ 52056 Aachen \\ Germany}
\address{JARA -- Fundamentals of Future Information Technology}
\email{melcher@rwth-aachen.de}

\author{Zisis N. Sakellaris}
\address{RWTH Aachen University \\ Lehrstuhl I f\"ur Mathematik  \\ 52056 Aachen \\ Germany}
\email{sakellaris@math1.rwth-aachen.de}

\begin{document}
\maketitle

\begin{abstract} We examine skyrmionic field configurations on a spherical ferromagnet with large normal anisotropy.
Exploiting variational concepts of angular momentum
we find a new family of localized solutions to the Landau-Lifshitz equation that are topologically distinct
from the ground state and not equivariant. Significantly, we observe an emergent spin-orbit coupling on the level of magnetization dynamics in a simple system without individual rotational invariance in spin and coordinate space. 
\end{abstract}

\section{Introduction}
We consider a spherical ferromagnet with perpendicular anisotropy described by magnetization fields $\m: \St \to \St$ governed by the energy 
\begin{equation} \label{eq:energy}
E(\m)= \frac{1}{2} \int_{\St} \left[ g^{\alpha \beta} \dd{\m}{x_\alpha}  \dd{\m}{x_\beta} + \kappa  \left(1-  (\m \cdot \n)^2 \right)\right] \, \mathrm{d} \sigma
\end{equation}
where $\n$ is the outer unit normal field to the $2$-sphere $\St$ and $\kappa>0$ is an anisotropy parameter. The energy is frame indifferent, i.e., we have invariance $E(\m_{R}) = E(\m)$ under joint rotations 
\begin{equation} \label{eq:joint_rotation}
\m_{R}(\bs y)=  R  \m(R^{-1} \bs y)
\end{equation}
for all $R \in \mathrm{SO}(3)$ and $\y \in \St$. Due to the link between anisotropy and curvature, however, there is no invariance with respect to individual rotations in spin and coordinate space, in contrast to conventional sigma or Skyrme models \cite{Belavin_Polyakov, Esteban, Lin_Yang, Manton_Sutcliffe}.
The reduced symmetry has a stabilizing effect on localized solitonic structures \cite{Krav} in analogy with chiral skyrmions in magnetic systems without inversion symmetry \cite{Bogdanov_Hubert_1994, Bogdanov_Yablonskii,  Mel1}. The topological classification of such field configurations $\m$ is based on the topological charge or skyrmion number $Q(\m) \in \Z$. For sufficiently regular $\m:\St \to \St$, it is given by
\begin{equation} \label{eq:charge}
Q(\m) = \frac{1}{4 \pi} \int_{\St} \m^\ast \omega_{\St} 
\end{equation}
where $\omega_{\St}$ is the standard volume form on $\St$ and $\m^\ast \omega_{\St}$ its pull back by $\m$.  For sufficiently large $\kappa$ the ground state is the hedgehog $\m=\pm \n$ as recently proven in \cite{Di_Fratta_2019}. While ground states belong to topological classes with $Q= \pm 1$, skyrmionic solutions with $Q=0$ emerge as excited states for large $\kappa$ even without the aid of chirality inducing spin-orbit terms \cite{Krav}. As $\kappa \to \infty$ such local minima subconverge in measure to a hedgehog while accumulating topological charge at a point on $\St$, the skyrmion center.  One intention of this letter is to give a rigorous variational footing to this observation based on the attainment of 
\[
\inf \{ E(\m):\m \in H^1(\St;\St) \text{  with  } Q(\m)=0\}<8 \pi
\]
for arbitrary large $\kappa$, where $H^1(\St;\St)$ is the energy space based on the usual Sobolev space $H^1(\St;\R^3)$.
The strict energy bound  (Lemma \ref{lemma:energy}) in the spirit of \cite{Brezis_Coron} expresses stability in the sense that collapsing part of the topological charge is energetically unfavorable. Such local energy minimizers are static solutions of the governing Landau-Lifshitz equation
\begin{equation} \label{eq:Landau_Lifshitz}
\del_t \m = \m \times \left[ \Delta_{\St}\m + \kappa  (\m \cdot \n) \n \right]
\end{equation}
where $\Delta_{\St}$ denotes the Laplace-Beltrami operator.
Aiming at a more general class of time periodic solutions of \eqref{eq:Landau_Lifshitz} it is customary to
take into account a further variational constraint or interaction term in form of a conserved functional that may be identified with a form of
angular momentum \cite{Goldstein_angular, Komineas_Papanicolaou, Kosevich, Papanicolaou_Tomaras, Papanicolaou_Zakrezewski, Yan_2013}. 
The angular momentum decomposes into a spin and orbital part accounting for rotation in spin and coordinate space, respectively. 
Their possible interplay is traditionally discussed in the presence of magnetostatic interactions and more recently in the presence of chiral
interactions \cite{Schuette_2014}. On the level of reduced rotational invariance, both situations are comparable with the present system.
Rigorous existence results to date address the spatially co-rotational case of dynamically stabilized magnetic bubbles with coherently precessing spins \cite{Gustafson_Shatah}. In this letter we focus on joint rotations and configurations of reduced rotational symmetry.  We introduce the following notion of angular momentum on $\St$
\begin{equation} \label{eq:angular}
\bs J(\m) = \bs S(\m) + \bs L(\m)\in \R^3
\end{equation}
decomposing into spin angular momentum
\begin{equation} \label{eq:spin}
{\bs S}(\m) =  \int_{\St} \m  \, d \sigma
\end{equation}
averaging $\m$ with respect to the surface measure,
and orbital angular momentum 
\begin{equation} \label{eq:orbital}
{\bs L}(\m)= \int_{\St} \n \, (\m^\ast \omega_{\St}) 
\end{equation}
averaging the outer unit normal $\n$ with respect to the topological charge distribution of $\m$. 
With an appropriate form of Poisson bracket \eqref{eq:Poisson}, the vectorial components of the individual momenta feature the usual commutation relations 
\begin{equation} \label{eq:commutation}
\{S_i, S_j\} = \epsilon_{ijk} S_k \quad \text{and} \quad \{ L_i, L_j\} =  \epsilon_{ijk} L_k. 
\end{equation}
While $\bs J$ is a conserved quantity of \eqref{eq:Landau_Lifshitz}, $\bs S$ and $\bs L$ are not individually conserved as a consequence
of the lack of individual rotational invariance of $E$. In fact, spin and orbital angular momenta are generators of spin and coordinate rotations, respectively, 
\begin{equation} \label{eq:generators}
\{\m, S_3\} =  \e_3 \times \m  \quad \text{and} \quad  \{\m, L_3\} = -\dd{\m}{\chi}
\end{equation}
where $\dd{}{\chi} = x \times \nabla_x$ is the angular derivative around the $\e_3$ axis and $x \in \R^2$ a stereographic coordinate centered at the poles. This motivates a variational approach to skyrmionic solutions 
$\m(\cdot, t)=  \m_{R(\nu t)}$ performing a joint rotation in spin and coordinate space for some angle velocity $\nu$. Profiles $\m=\m(x)$ with $\bs J(\m)=J_0 \e_3$ of such spinning solutions may be obtained via a constrained variational principle 
\begin{equation} \label{eq:J_inf}
\inf \{ E(\m):\m \in H^1(\St;\St) \text{  with  } Q(\m)=0 \text{  and  } \bs J(\m)= \bs J_0 \}
\end{equation}
with $\nu$ appearing as Langrange multiplier. Our chief result is a proof of attainment.
 \begin{theorem}
For every $\kappa>0$ there exists $\eps>0$ such that if $\bs J_0 \in \R^3$ satisfies the bounds 
$4\pi <|\bs J_0| < 4\pi + \eps$,  then \eqref{eq:J_inf} is attained by a smooth 
field $\m:\St \to \St$ with $E(\m)< 8 \pi$ which is not equivariant.  Moreover, $\m$ is the profile of a jointly rotating solution
of the Landau-Lifshitz equation \eqref{eq:Landau_Lifshitz}. 
\end{theorem}

Equivariant fields $\m$ possess, by definition, a rotation axis such that $\m_R=\m$ for all $R \in \mathrm{SO}(3)$ in the stabilizer of this axis. The value $|\bs J| = 4\pi$ is critical in the sense that it is assumed by all equivariant fields of degree $Q=0$ (Lemma \ref{lemma:equivariant}), which are precisely the critical fields of $\bs J$ in this topological class. The angular momentum therefore serves as
a measure of rotational symmetry, and fields of interest necessarily exhibit reduced rotational symmetry. Examining second variations of $\bs J$ (Lemma \ref{lemma:elliptical}), we shall show that elliptical distortions of equivariant fields near the energy minimum strictly increase the size of angular momentum. As usual \cite{Brezis_Coron, Esteban, Lin_Yang, Mel1} the key to existence is an energy estimate (Lemma \ref{lemma:energy}) that yields compactness of suitable minimizing sequences.\\

An important open problem is to understand how minimal energies depend on the size of the angular momentum in order to 
estimate the rotation frequency $\nu$ and to ascertain the obtained solution to the Landau-Lifshitz equation is non-static. 
This is of course closely related to the notoriously difficult problem of proving symmetry of minimizing skyrmions. 
Carrying out a similar analysis for chiral skyrmions in the non-compact space $\R^2$ as in \cite{Mel1} is substantially more challenging since
the orbital angular momentum functional based on the second moment of the topological charge density \cite{Papanicolaou_Tomaras}
may become unbounded for finite energy configurations. Another interesting perspective is to examine different forms of dynamic excitation or stabilization of skyrmions combining precession in spin and breathing in coordinate space as suggested in \cite{Schuette_2014, Zhou_2015}.

\section{Notation and preliminaries} 
\subsection*{Stereographic representation} We equip $\St \subset \R^3$ with the orientation given by the outer unit normal $\n$.
Most arguments are carried out in orientation preserving stereographic coordinates $x \in \R^2$ centered at the north pole. 
Points $\y \in \St \setminus \{ -\e_3\}$ are parametrized by
\begin{equation} \label{eq:stereographic}
\bs \Phi(x) = \lambda(x) \left(x, \frac{1-|x|^2}{2}  \right) 
\quad \text{with conformal factor} \quad \lambda(x)=\frac{2}{1+|x|^2}.
\end{equation} 
The metric tensor and the surface element are given by 
$g(x)= \lambda^2(x) \bs 1$
and $\mathrm{d} \sigma= \lambda^2(x) \mathrm{d} x$, respectively, so that the energy \eqref{eq:energy} may be written
\[
E(\m)= \frac 1 2  \int_{\R^2}\left[  |\nabla \m|^2 + \kappa \left( 1-(\m \cdot \n)^2 \right)  \lambda^2(x) \right] \, \mathrm{d} x.
\]
Here and in the following we use the same notation for $\m:\St \to \St$ and its pull back $\m:\R^2 \to \St$ by $\bs \Phi$.
By conformal invariance, the exchange or Dirichlet energy is not distinguishable from the flat case. 
The skyrmion number \eqref{eq:charge} may be expressed in terms of the topological vorticity
\begin{equation} \label{eq:omega}
 \omega(\m) = \m \cdot \left( \frac{\del \m}{\del x_1} \times  \frac{\del \m}{\del x_2} \right) 
\end{equation}
through the identity $\m^\ast \omega_{\St} = \omega(\m) \mathrm{d} x_1 \wedge \mathrm{d} x_2$ as
\[
Q(\m) = \frac{1}{4 \pi} \int_{\R^2} \omega(\m) \, \mathrm{d} x.
\]
The differential form definition of the orbital angular momentum \eqref{eq:orbital} can be turned into a surface 
integral in terms of the surface Jacobian $\mathcal{J}_{\m}$ of $\m$, i.e., the Hodge dual of 
$\m^\ast \omega_{\St} =  \mathcal{J}_{\m} d \sigma$ leading to
${\bs L}(\m)= \int_{\St} \n \, \mathcal{J}_{\m}\, \mathrm{d} \sigma$.
In stereographic coordinates \eqref{eq:stereographic} the total angular momentum \eqref{eq:angular} reads
\begin{equation}
{\bs J}(\m)= \int_{\R^2}  \lambda^2(x) \, \m + \bs \Phi(x) \omega(\m) \, \mathrm{d} x.
\end{equation}
In the last component we obtain in particular
\[
S_3(\m)= \int_{\R^2} \lambda^2(x) \, m_3 \mathrm{d} x 
\]
and 
\begin{equation} \label{eq:stereo_L}
L_3(\m) =  4 \pi Q(\m) - \int_{\R^2} \lambda(x)|x|^2   \omega(\m) \, \mathrm{d} x.
\end{equation}
For $|x| \ll 1$ the orbital part reduces to the second moment definition from \cite{Komineas_Papanicolaou, Papanicolaou_Tomaras}.

\subsection*{Landau-Lifshitz equation}
 Returning to the energy, the variational $L^2$-gradient with respect to the surface measure on $\St$ will be denoted by $\bs \nabla E(\m)$ and is defined by the relation
\[
\int_{\St} \bs \nabla E(\m) \cdot \bs \phi \, \mathrm{d} \sigma = \delta E(\m) \langle \bs \phi \rangle = \frac{\mathrm{d}}{\mathrm{d}s}\Big|_{s=0} E(\m_s) 
\]
for variations $\m_s$ with $\m_s = \m + s  \bs \phi + o(s)$ with $\bs \phi \in C^\infty(\St; \R^3)$. Explicitly  
\[
\bs \nabla E(\m)= - \left[ \Delta_{\St} \m + \kappa  (\m \cdot \n) \n \right],
\]
where $\Delta_{\St} = \lambda^{-2} \Delta_{x}$ is the Laplace-Beltrami operator.
We are interested in static and dynamic solutions of the Landau-Lifshitz equation 
\[
\del_t \m + \m \times \bs \nabla E(\m) =0
\]
which is the abstract form of \eqref{eq:Landau_Lifshitz}.
In the sequel we shall use the Poisson bracket for functionals $F$ and $G$ given by 
\begin{equation} \label{eq:Poisson}
\{F,G\} =  \int_{\St} \m \cdot \left[ \bs \nabla F(\m) \times \bs \nabla G(\m) \right] \mathrm{d} \sigma.
\end{equation}
In the weak formulation $\{F,G\} = - \delta F(\m) \left\langle  \m \times  \bs \nabla G(\m) \right \rangle$, 
the evaluation functional $F(\m) = \m$ may be included,  i.e. $\{\m, E\}= - \m \times \bs \nabla E(\m)$, giving rise to the
Hamiltonian formulation of the Landau-Lifshitz equation
\[
\del_t \m = \{\m, E\}.
\]

\section{Symmetry and conservation of angular momentum}

\begin{lemma} \label{lemma:frame_indifference}
The angular momentum is frame indifferent in the sense that 
\[
\bs J(\m_R) = R \bs J(\m) \quad \text{for all} \quad R \in \mathrm{SO}(3).
\]
\end{lemma}
\begin{proof} In fact, the individual momenta $\bs S$ and $\bs L$ are frame indifferent, which is obvious for $\bs S$. 
Concerning $\bs L$ we may use 
differential form calculus. In terms of the orientation preserving transformation $T(\y)=R \y$ we have $\m_R= T \circ \m \circ T^{-1}$ and
$\m_R^\ast\omega_{\St} = (T^{-1})^\ast(\m^\ast\omega_{\St})$ using rotational invariance of $\omega_{\St}$. 
With $(T^\ast \n)=R \n$, it
 follows that
\[
\bs L(\m_R) = \int_{\St}  \n \, (T^{-1})^\ast(\m^\ast\omega_{\St})
= \int_{\St}   (T^{-1})^\ast \left( R \n \, (\m^\ast\omega_{\St}) \right) = R \bs L(\m)\]
by the invariance property of the form integral.
\end{proof}

\begin{lemma} \label{lemma:1st_var}
For smooth tangent vector fields $\bs \phi \perp \m$ 
\[
\delta S_3(\m)\langle \bs \phi \rangle = \int_{\St}  \e_3  \cdot \bs \phi \, d \sigma
\quad
\text{and} 
\quad
\delta L_3(\m)\langle \bs \phi \rangle = - \int_{\St} \left( \m \times\dd{\m}{\chi}  \right) \cdot \bs \phi \, d \sigma 
\]
where $\dd{}{\chi}= x \times \nabla_x $ denotes the angular derivative around $\e_3$.
\end{lemma}

\begin{proof} We only compute the variation of $L_3$. 
Taking into account 
\begin{equation} \label{eq:density}
\nabla_x (\lambda(x) |x|^2) =\lambda^2(x) x
\end{equation}
we obtain for smooth $\bs \phi$ not necessarily tangent 
\begin{eqnarray}\label{eq:delta_L}
\delta L_3(\m)\langle \bs \phi \rangle = &-& \int_{\R^2} \left(\m \times\dd{\m}{\chi}\right) \cdot \bs \phi   \lambda^2(x) \mathrm{d} x  \nonumber \\
 &-& 3 \int_{\R^2}   \lambda(x) |x|^2 \left( \frac{\del \m}{\del x_1} \times \frac{\del \m}{\del x_2} \right) \cdot \bs \phi  \, \mathrm{d} x,
\end{eqnarray}
where the second integral vanishes if $\bs \phi$ is tangent to $\m$. 
\end{proof}

Expressed in terms of the Poisson bracket \eqref{eq:Poisson}, spin and orbital angular momenta are now seen to emerge as generators of spin and space rotations around the $\e_3$ axis, respectively, as pointed out in \eqref{eq:generators}. Hence
\begin{equation} \label{eq:Poisson_J}
\{ \m, J_3 \} =  \e_3 \times  \m-\dd{\m}{\chi}.
\end{equation}
With the notation \eqref{eq:joint_rotation} for joint rotations we are looking for time periodic solutions of the Landau-Lifshitz equation \eqref{eq:Landau_Lifshitz} of the form
\begin{equation} \label{eq:rotating}
\m(\cdot, t)=  \m_{R(\nu t)} \quad \text{with rotations} \quad R(\alpha) = \exp(\alpha \, \mathbb{J}) 
\end{equation}
where $\mathbb{J}\y= \e_3 \times\y$.
A straight forward calculation taking into account \eqref{eq:Poisson_J} yields:
\begin{lemma}  \label{lemma:rotation}
For all $\m:\St \to \St$ the following identity holds true
\[
\frac{\del}{\del t} \m_{R(\nu t)}  = \nu \{\m, J_3\}_{R(\nu t)}.
\]
In particular, if $\{\m, E\} = \nu \{\m, J_3\}$ then $\m_{R(\nu t)}$ is a solution of \eqref{eq:Landau_Lifshitz}.
\end{lemma}

\begin{lemma} \label{lemma:conservation}
$\bs J(\m)$ is conserved for smooth solutions $\m=\m(t)$ of \eqref{eq:Landau_Lifshitz}.
\end{lemma}

\begin{proof} By frame indifference of the Landau-Lifshitz equation, it suffices to prove conservation of $J_3$. 
From \eqref{eq:Landau_Lifshitz} we immediately obtain
\[
\frac{d}{dt} S_3(\m) = \frac{d}{dt}  \int_{\St} m_3  \, d \sigma   
 =    \kappa  \int_{\St} (\m \cdot \n) (\m \times \n)_3 \, \mathrm{d} \sigma.
\]
Moreover in stereographic coordinates \eqref{eq:stereographic}
\[
\nabla \m \cdot (\m \times \del_t \m) = - \nabla \m \cdot \left( \lambda^{-2} \Delta \m + \kappa  (\m \cdot \n) \n \right)  
\]
and taking the curl ($\nabla \times $) we obtain the conservation law \cite{KMMS2011}
\begin{equation} \label{eq:omega_conservation}
\dd{}{t} \omega(\m) =  \curl \left( \kappa   (\m \cdot \n) (\nabla  \n \cdot \m)-  \lambda^{-2} \div (\nabla \m \otimes \nabla \m)  \right).
\end{equation}

Since $Q(\m)$ is conserved 
we obtain from \eqref{eq:omega_conservation} after integration by parts
\begin{eqnarray*}
 \frac{d}{dt} L_3(\m) &=& - \frac{d}{dt} \int_{\R^2}   \lambda(x) |x|^2 \omega(\m) \, \mathrm{d} x \\
&=& \int_{\R^2} \left[  \kappa  (\m \cdot \n) \left(\dd{\n}{\chi} \cdot \m\right) \lambda^2(x)+\mathbb{J}: (\nabla \m \otimes \nabla \m)   \right]  \, \mathrm{d} x \\
&=& \kappa  \int_{\St}  (\m \cdot \n) \left(\dd{\n}{\chi} \cdot \m\right) \mathrm{d} \sigma = - \kappa  \int_{\St}  (\m \cdot \n) (\m \times  \n)_3 \mathrm{d} \sigma. 
\end{eqnarray*}
In fact, $\mathbb{J}=- \nabla_x \times x$ comes about by taking into account \eqref{eq:density}, and the corresponding term is dropping out
by skew-symmetry. In the last step we used $\frac{\del \n}{ \del \chi} = \e_3 \times \n$. It follows that $\frac{d}{dt} J_3(\m)=\frac{d}{dt} (S_3+L_3)(\m)=0$.
\end{proof}

\section{Role of equivariance and elliptic distortion} 

In the following discussion, frame indifference (Lemma \ref{lemma:frame_indifference}) allows us to fix the canonical rotation axis $\e_3$. It is customary to represent the stereographic coordinate as $x = r e^{\mathrm{i} \chi}$ and the magnetization field
as $\m = (\sqrt{1-m_3^2} \, e^{\rm{i} \varphi}, m_3)$ so that
\begin{equation} \label{eq:omega_polar}
\omega(\m) = \frac 1 r  \left(\dd{\varphi}{r}  \dd{m_3}{\chi} -  \dd{m_3}{r} \dd{\varphi}{\chi}  \right).
\end{equation}
Magnetic chiral skyrmions have originally been found in the class of co-rotational (or axisymmetric) fields where  $m_3=m_3(r)$ and $\varphi(\chi)=\chi + \alpha$ for a well-defined phase shift $\alpha$ depending on the form of antisymmetric exchange interaction \cite{Bogdanov_Hubert_1994, Bogdanov_Hubert_1999}.  
In a deep ferromagnetic regime such solutions are indeed locally energy minimizing \cite{Li_Melcher}.
Precessional Landau-Lifshitz dynamics and angular momentum \eqref{eq:angular} are rather linked to the notion of equivariance, which summarizes a wider class of fields $\m:\R^2 \to \St$ satisfying $\m_R=\m$ for $R \in \mathrm{SO}(2)\subset \mathrm{SO}(3)$.
More generally $\m(x) = R^k \m(R^{-1}x)$ for some $k \in \Z$ is referred to as $k$-equivariance. This can also be expressed as
\[
m_3=m_3(r) \quad  \text{and} \quad \varphi(r, \chi)= k \chi+\alpha(r).
\]
From \eqref{eq:omega_polar}
\begin{equation} \label{eq:Q_polar}
Q(\m) = \frac{k}{2}  \left( m_3(\e_3) -  m_3(-\e_3) \right).
\end{equation}
Hence $k$-equivariant fields $\m$ with $Q(\m)=0$ are characterized by non-zero
polarity
\[
p(\m)= \frac{1}{2} \left(m_3(\e_3)+  m_3(-\e_3) \right).
\]

\begin{lemma} \label{lemma:equivariant}
For $k$-equivariant fields $\m$ it holds that $J_1(\m)= J_2(\m)=0$ and
\[
J_3(\m)=   (1-k) S_3(\m) +4 \pi k p(\m).
\]
In particular, $|\bs J|= 4\pi |p|$ for $1$-equivariant fields, the stationary points of $J_3$.
\end{lemma}

\begin{proof} The claim that $J_1, J_2=0$ follows upon integration in $\chi$.
By virtue of \eqref{eq:omega_polar} and integration by parts
\begin{eqnarray*}
L_3(\m) &=& 4 \pi Q(\m) +  2\pi k \, \int_0^\infty  r^2 \lambda(r) \dd{ m_3 }{r}  dr \\
&=& 4 \pi \left( Q(\m) + k  m_3(\infty) \right) - 2 \pi k  \int_0^\infty   \lambda^2(r)  m_3 \,   r dr
\end{eqnarray*}
which implies the formula for $J_3$ by \eqref{eq:Q_polar}.
Finally
\[
\del_\chi \m = \e_3 \times \m \quad \text{iff} \quad   m_3= m_3(r) \text{  and  } \varphi(r, \chi)= \chi+\alpha(r)
\]
which implies the final claim by virtue of Lemma \ref{lemma:1st_var}.
\end{proof}

Let us now examine the effect of elliptical distortions.
\begin{lemma} \label{lemma:elliptical}
Suppose $\m$ is equivariant and
$\m_s(x)=\m(sx_1,x_2)$ for $s >0$ and $x \in \R^2$. Then 
$
J_1(\m_{s})= J_2(\m_{s})=0
$ 
and
\[
\frac{\mathrm{d}^2}{\mathrm{d}s^2}\Big|_{s=1} \,J_3(\m_{s})  = - \frac{1}{2} \int_{\R^2}   \omega(\m) |x|^2 \lambda^2(x) \, \mathrm{d} x.
\]
\end{lemma}

Since $\omega(\n)=\lambda^2(x)>0$ for the map $\n=\n(\bs{\Phi}(x))$, Lemma \ref{lemma:equivariant} and \ref{lemma:elliptical} imply:
\begin{corollary}\label{cor:maximum}
In the class of almost minimizing co-rotational skyrmions such that $Q(\m)=0$, $E(\m)< 8 \pi$ and $\m=\n$ away from a small spherical cap centered at the north pole $x=0$,
$J_3(\m_s)$ has a strict local maximum $J_3(\m)=-4\pi$ at $s=1$. In other words,
the size of $\bs{J}(\m_s)$ has a strict local minimum at $s=1$.
\end{corollary}

\begin{proof}[Proof of Lemma \ref{lemma:elliptical}] The claim that $J_1, J_2=0$ follows easily from reflection arguments, taking into account Lemma \ref{lemma:frame_indifference}. 
For the claim about $J_3$ we observe that (writing $J=J_3$)
\[
\frac{\mathrm{d}^2}{\mathrm{d}s^2}\Big|_{s=1} \,J(\m_{s})  
= {\rm Hess} \, J(\m) \left\langle  \dot{\m} ,\dot{\m}  \right\rangle
\quad \text{with} \quad \dot{\m} = \frac{\mathrm{d}}{\mathrm{d}s}\Big|_{s=1} \m_s 
\]
where for admissible tangent fields $\bs \phi \perp \m$, the Hessian is given by
\[
{\rm Hess} \, J(\m) \langle \bs \phi, \bs \phi \rangle = \delta^2J(\m) \langle \bs \phi, \bs \phi \rangle - \delta J(\m)\langle |\bs \phi|^2 \m \rangle. 
\]
From Lemma \ref{lemma:1st_var} and taking into account \eqref{eq:delta_L} we obtain
\[
 \delta J(\m)\langle |\bs \phi|^2 \m \rangle   =    \int_{\St } m_3   |\bs \phi|^2 \mathrm{d} \sigma
 - 3 \int_{\R^2}   \lambda(x) |x|^2 \omega(\m) \, |\bs \phi|^2  \, \mathrm{d} x,
\]
and in turn
\begin{eqnarray*}
 \delta^2 J(\m)\langle\bs \phi,  \bs\phi \rangle = 
 &-&  \int_{\St}   \left( \m \times \dd{\bs \phi }{\chi} \right) \cdot \bs \phi \, \mathrm{d} \sigma \\
 &+& 3 \sum_{\alpha, \beta =1,2} \epsilon_{\alpha \beta}   \int_{\R^2}   \lambda(x) |x|^2 \left(  \bs \phi  \times \frac{\del \m}{\del x_\beta} \right) \cdot  \frac{\del \bs \phi}{\del x_\alpha} \, \mathrm{d} x.
\end{eqnarray*}
Letting $\bs \phi = \dot{\m}= x_1 \dfrac{\del \m}{\del x_1}$ it follows from 
\[
 \left( \frac{\del \m}{\del x_1}  \times \frac{\del \m}{\del x_2} \right) \cdot  \frac{\del^2 \m}{\del x_1^2} 
 = \omega(\m) \; \m \cdot  \frac{\del^2 \m}{\del x_1^2} = - \omega(\m)  \left|\frac{\del \m}{\del x_1}\right|^2
\]
that the second integrals in $\delta J$ and $\delta^2 J$ cancel and
\[
\frac{\mathrm{d}^2}{\mathrm{d}s^2}\Big|_{ s=1} \,J(\m_{s}) = - \int_{\R^2} x_1^2  \left( \left( \m \times \frac{\del}{\del \chi} \left( \frac{\del \m}{\del x_1} \right)\right) \cdot  \frac{\del \m}{\del x_1} + m_3   \left|\frac{\del \m}{\del x_1}\right|^2 \right) \lambda^2(x) \mathrm{d} x.
\]
Since
$
\frac{\del}{\del \chi} \left(  \frac{\del \m}{\del x_1} \right) =  \frac{\del}{\del x_1} \left( \frac{\del \m}{\del \chi} \right)  - \frac{\del \m}{\del x_2}$
and by equivariance $\dd{\m}{\chi}= \e_3 \times \m$
\[
  \left( \m \times \frac{\del}{\del \chi} \left( \frac{\del \m}{\del x_1} \right)\right) \cdot  \frac{\del \m}{\del x_1} =
  \omega(\m) - m_3   \left|\frac{\del \m}{\del x_1}\right|^2,\]
the asserted formula for $\frac{\mathrm{d}^2}{\mathrm{d}s^2}\big|_{s=1} \,J(\m_{s})$ immediately follows taking into account the radial symmetry of $\lambda$ and $\omega(\m)$.
\end{proof}

\section{Energy bounds and attainment}

In this section we give a proof of the main Theorem stated in the introduction, using methods from the calculus of variations.
To this end, it is routine to extend the functionals $E$, $Q$ and $\bs J$ continuously to the energy space $H^1(\St;\St)$, see
e.g. \cite{Mel1}. By Lemma \ref{lemma:frame_indifference} we may assume $\bs J_0=J_0 \e_3$.
Once attainment of \eqref{eq:J_inf} by a smooth field is established, the remaining assertions follow from Lemma \ref{lemma:equivariant},
and from the Langrange multiplier theorem and Lemma \ref{lemma:rotation}, respectively.

\begin{lemma}\label{lemma:energy}
For every $ \kappa>0$ there exists a co-rotational field $\m \in H^1(\St;\St)$ with $Q(\m)=0$ and $E(\m)<8\pi.$
\end{lemma}

\subsection*{Moving frame formulation} We shall construct a suitable trial field $\m:\R^2\to \St$ in stereographic coordinates 
such that $\m(x)= \n(\bs \Phi(x))$ for large $x$. 
The construction will be based on a coordinate field $\bs u: \R^2 \to \St$ that represents $\m$ in the orthogonal frame  $\{\taf_1, \taf_2, \n \}$
where the $\taf_{\alpha}$ are the unit coordinate vector fields for $\bs \Phi$. In this framework the normal anisotropy simplifies
to the conventional uniaxial anisotropy.  The exchange energy expands into several parts, which can be conveniently determined by virtue of  Cartan's calculus as in \cite{Krav}. One obtains that $E(\m)=\mathcal{E}_0(\bs u) +\mathcal{E}_1(\bs u)$
where
\[
\mathcal{E}_0(\bs u)   =  \int_{\R^2} \left[ \dfrac{1}{2} |\nabla \bs u|^2+ 
 \big[   u_3 (\nabla \cdot  u) - ( u \cdot \nabla) u_3   \big] \lambda(x) +  \dfrac{\kappa}{2} (1- u_3^2) \lambda^2(x) \right] \mathrm{d}x 
\]
and
\[
\mathcal{E}_1(\bs u) = \int_{\R^2} \left[  \left(  (1- u_3^2) - \left(u \times \frac{\del u}{ \del \chi}\right)  \right)  \lambda(x) + \left(u_3^2 -  (x \cdot u)u_3 \right)\lambda^2(x)  \right] \mathrm{d}x.
\]
As already observed in \cite{Krav}, $\mathcal{E}_0$ is a slight modification of the energy of a chiral magnet which interfacial Dzyaloshinskii-Moriya interaction, revealing the stabilizing effect of curvature of the sphere. Our construction will be based on a co-rotational Lipschitz map $\bs u$ with anti-conformal core and finite tail, following the argument in \cite{Mel1}.  

\begin{proof}[Sketch of proof of Lemma \ref{lemma:energy}] Taking into account the identity \cite{Hoffmann, Mel1} 
\[
 \dfrac{1}{2} |\nabla \bs \Phi|^2 =
 \big[   \bs \Phi (\nabla \cdot  \bs \Phi) - ( \bs \Phi \cdot \nabla) \bs \Phi   \big]_3 = \lambda^2
\]
for the stereographic map \eqref{eq:stereographic}
we set $\bs u(x) = \bs \Phi_{-}(x/\eps)$ for $0\le |x| \le 1$ and some scale $\eps>0$ depending on $\kappa$.
Here $\boldsymbol{\Phi}_{-}$ is the orientation reversing stereographic map obtained by reversing the sign of $\Phi_3$
with polar profile $\theta_0(r)= \pi -2\arctan (r)$. We extend $\bs u$ in a co-rotational manner by setting $u_3= \cos \theta$  where $\theta(r)=\theta_0 (1/\eps)(2-r)$ for $1<r<2$ and zero else. We obtain as in \cite{Mel1} 
\[
\mathcal{E}_0(\bs u) \le 4 \pi \left(1-\eps \right) + o_{\kappa}(\eps)
\]
as $\eps \to 0$. Since $1- u_3^2 = u \times \dd{u}{\chi}$, it follows that
\[
\mathcal{E}_1(\bs u)  \le 4 \pi +\mathcal{E}_2(\bs u) 
\quad \text{where} \quad
\mathcal{E}_2(\bs u) = -\intt  \left(  |x| \left(\sqrt {1 -u_3^2}\right)u_3  \right)\lambda^2(x) \mathrm{d}x.
\]
Using that $1 \le \lambda^2(x) \le 4$ for all $|x| \le 1$ we obtain
\begin{align*}
 & \mathcal{E}_2^{\rm{core}}(\bs u)= 2 \eps^3 \int_{|x| \le 1/\eps} \dfrac{ |x|^2(1 - |x|^2 )}{(1+|x|^2)^2} \lambda^2(\eps x) \mathrm{d}x  \\
& \leq 4\pi \eps^{3}  \left(  4 \int_0^1 \dfrac{ r^3(1- r^2 )}{(1+r^2)^2} \mathrm{d}r +   \int_1^{1/\eps}\dfrac{r^3(  1- r^2  )}{(1+r^2)^2}  \mathrm{d}r   \right) = - 2\pi \eps + o(\eps) 
\end{align*}
as $\eps \to 0$. Since $u_3(r) \ge 0$ for $1 \le r \le 2$ and $\eps$ sufficiently small we have $\mathcal{E}_2^{\rm{tail}}(\bs u)\le 0$.
Hence $\mathcal{E}_0(\bs u)+\mathcal{E}_1(\bs u) \le  8 \pi - 6 \pi \eps   + o_{\kappa}(\eps)$ as $\eps \to 0$, and the estimate follows. By construction, the map $\m=u_1 \taf_1 + u_2 \taf_2 + u_3 \n$ has the same energy. Finally, $-\e_3$ is a regular value
of $\m$, which is attained precisely at the poles near which the map is smooth and orientation reversing and preserving, respectively. Using a localization argument as in \cite{KMMS2011} Proposition 2, Brouwer's degree formula yields $Q(\m)=0$.
\end{proof}

The construction interpolates between the hedgehog map $\m = \n$ near the south pole and and its strongly localized  inversion near the north pole. Rescaling $\bs u$ or following the strategy in \cite{Doering_Melcher}, the construction may be modified in such a way that $\m=\n$ (i.e. $\bs u =\e_3$) away from an arbitrarily small spherical cap.
Since the energy is continuous in $H^1(\St; \St)$, an elliptical distortion in the spirit of  Corollary \ref{cor:maximum}
yields the following extension. 

\begin{corollary} \label{cor:energy}
For every $\kappa>0$ there exists $\eps>0$ such that if 
$4\pi < |J_0| < 4\pi+\eps$, then there exists $\m \in H^1(\St;\St)$ with $Q(\m)=0$, $\bs J(\m)= J_0 \e_3$,
and $E(\m)<8\pi$.
\end{corollary}

\subsection*{Proof of attainment} 
Let for an admissible set of parameters
\[
E_\ast= \inf \{ E(\m):\m \in H^1(\St;\St) \text{  with  } Q(\m)=0 \text{  and  } \bs J(\m)= J_0 \e_3\}
\]
and $\{\m_{n}\}_{n\in \N}\subset H^{1}(\St, \St)$  be a minimizing sequence with $\lim_{n \to \infty} E(\m_n) = E_\ast$ such that $Q(\m_{n})=0$
and $\bs J(\m_n)=J_0 \e_3$ for all $n\in\N$. Then Corollary \ref{cor:energy} implies that $E(\m_{n})<8\pi$. 
Passing to a subsequence we may assume $\m_n \to \m$ in $L^2$ and weakly in $H^1$ as $n \to \infty$ for some $\m\in H^1(\St;\St)$
such that $E(\m)\leq\liminf_{n\to\infty}E(\m_n)=E_\ast$ and $\bs S(\m) = \lim_{n\to\infty}\bs S(\m_n)$. 
It remains to verify that $Q$ and $\bs J$ behave continuously, which is true provided $\omega(\m_{n}) \weak \omega( \m)$ 
appropriately as $n \to \infty$. If not, then, by virtue of \cite{BCL} Theorem E.1 and \cite{Lions} Lemma 4.3,  there exist 
finitely many points $x_1\ldots x_N\in \R^2 \cup \{\infty\}$ and $q_1,\ldots, q_N\in  \Z \setminus \{0\}$ satisfying
\begin{equation} \label{eq:cc}
E(\m) + 4 \pi \sum_{i=1}^N |q_i| \le E_\ast \quad \text{and} \quad Q(\m) +  \sum_{i=1}^N q_i =0 
\end{equation}
such that for a subsequence $n \to \infty$
\[
\omega(\m_{n}) \weak \omega( \m) + 4\pi  \sum_{i=1}^N q_i \delta_{x_i}
\]
weakly in the sense of measures on the compactified space $\R^2 \cup \{\infty\}$. Alternatively one may use a second orientation preserving stereographic chart. 
But as $E_\ast<8 \pi$
we  must have $N=1$ and $q_1 = \pm 1$. Hence $Q(\m)= \pm 1$ and $E(\m) \ge 4\pi$ by the classical topological lower bound \cite{Belavin_Polyakov}, contradicting \eqref{eq:cc}. Smoothness follows from well-established methods from the regularity theory 
of harmonic maps from surfaces, see e.g. \cite{Moser}.\\

The same line of arguments also displays the skyrmionic character of such field configurations in the regime of large $\kappa$
in terms of a compactness result: for $\kappa \to \infty$ and $\m_{\kappa} \in H^1(\St; \St)$ with $Q(\m_{\kappa})=0$ and 
$E_{\kappa}(\m_{\kappa})< 8 \pi$ we have subconvergence
\[
\m_{\kappa} \to \pm \n \quad \text{in measure and} \quad \omega(\m_{\kappa}) \weak \omega( \m)  \mp 4\pi  \delta_{x_0}.
\]
weakly in the sense of measures for some $\y_0=\bs \Phi(x_0) \in \St$.\\

\subsection*{Acknowledgements} 
We are indebted to Stavros Komineas for pointing out the relevance of angular momenta in the context of chiral magnetism and for valuable discussions on the subject matter. 

\bigskip
 
On behalf of all authors, the corresponding author states that there is no conflict of interest.  
  
\bibliography{Melcher_Sakellaris_LMP.bib}
\bibliographystyle{acm}

\end{document}